\documentclass[12pt,oneside,peerreview,draftcls,onecolumn]{IEEEtran}

\usepackage{amsmath,amsfonts,amssymb,amsbsy, amsthm,ae,aecompl}
\usepackage{algorithm,algpseudocode}
\usepackage[utf8]{inputenc}
\usepackage[english]{babel}
\usepackage{bm}
\usepackage{color}
\usepackage[noadjust]{cite}
\usepackage{epsfig}
\usepackage{enumerate}
\usepackage{float} 
\usepackage{fancyhdr}
\usepackage[T1]{fontenc}
\usepackage[acronym,toc,shortcuts]{glossaries}
\usepackage{graphicx, caption, subcaption}
\usepackage{hyperref}
\usepackage{lastpage}
\usepackage{listings}
\usepackage{lipsum}
\usepackage{dsfont}
\usepackage{multirow,tabularx}
\usepackage[normalem]{ulem}
\usepackage{transparent}
\usepackage{tikz,pgfplots}
\usepackage{times}
\usepackage{verbatim}
\usepackage[all]{xy}
\usepackage{verbatim}
\usetikzlibrary{calc}

\hypersetup{
    colorlinks,%
    citecolor=black,%
    filecolor=black,%
    linkcolor=black,%
    urlcolor=black
}

\pgfplotsset{
  grid style = {
    dash pattern = on 0.025mm off 0.95mm on 0.025mm off 0mm, 
    line cap = round,
    black,
    line width = 0.5pt
  },
  tick label style={font=\small},
  label style={font=\small},
  legend style={font=\footnotesize},
}

\pgfplotscreateplotcyclelist{laneasIACaching1}{
cyan!60!black,		solid, 	every mark/.append style={fill=cyan!60!black},mark=x\\%
cyan!60!black,		solid, 	every mark/.append style={fill=cyan!60!black},mark=+\\%
cyan!60!black,		solid, 	every mark/.append style={fill=cyan!60!black},mark=o\\%
red!80!black,       dashed\\%
cyan!60!black, 		densely dotted,	every mark/.append style={fill=cyan!80!black},mark=diamond*\\%
}

\graphicspath{{figures/}}

\makeglossaries
\newacronym{BS}{BS}{base station}
\newacronym{CDN}{CDN}{content delivery network}
\newacronym{CDF}{CDF}{cumulative distribution function}
\newacronym{CF}{CF}{collaborative filtering}
\newacronym{CN}{CN}{core network}
\newacronym{CRP}{CRP}{{C}hinese restaurant process}
\newacronym{CS}{CS}{central scheduler}
\newacronym{CSI}{CSI}{channel state information}
\newacronym{D2D}{D2D}{device-to-device}
\newacronym{DoF}{DoF}{degree-of-freedom}
\newacronym{HetNet}{HetNet}{heterogeneous network}
\newacronym{FDD}{FDD}{frequency-division duplex}
\newacronym{ICIC}{ICIC}{inter-cell interference coordination}
\newacronym{ICN}{ICN}{information-centric network}
\newacronym{IA}{IA}{interference alignment}
\newacronym{ISI}{ISI}{inter-stream interference}
\newacronym{IUI}{IUI}{inter-user interference}
\newacronym{LTE}{LTE}{long term evolution}
\newacronym{MIMO}{MIMO}{multiple-input multiple-output}
\newacronym{PPP}{PPP}{{P}oisson point process}
\newacronym{PHY}{PHY}{physical layer}
\newacronym{SBS}{SBS}{small base station}
\newacronym{SINR}{SINR}{signal-to-interference-plus-noise ratio}
\newacronym{SNR}{SNR}{signal-to-noise ratio}
\newacronym{SCN}{SCN}{small cell network}
\newacronym{SVD}{SVD}{singular value decomposition}
\newacronym{TL}{TL}{transfer learning}
\newacronym{TDD}{TDD}{time-division duplex}
\newacronym{UT}{UT}{user terminal}
\newacronym{QoS}{QoS}{quality-of-service}
\newacronym{QoE}{QoE}{quality-of-experience}
\newacronym{RAN}{RAN}{radio access network}

\newtheorem{theorem}{Theorem}
\newtheorem{lemma}{Lemma}
\newtheorem{proposition}{Proposition}
\newtheorem{remark}{Remark}

\begin{document}
\title{On the Benefits of Edge Caching for MIMO Interference Alignment
}
\author{
		\IEEEauthorblockN{Matha Deghel$^{\star,\diamond}$, Ejder Baştuğ$^{\diamond}$, Mohamad Assaad$^{\star}$ and Merouane Debbah$^{\diamond,\dagger}$} \\
		
		\IEEEauthorblockA{
				\small
				$^{\star}$Laboratoire de Signaux et Systèmes (L2S, UMR8506) CentraleSupélec-CNRS-Université Paris-Sud, Gif-sur-Yvette, France \\				
				$^{\diamond}$Large Networks and Systems Group (LANEAS), CentraleSupélec, Gif-sur-Yvette, France \\	
				$^{\dagger}$Mathematical and Algorithmic Sciences Lab, Huawei France R\&D, Paris, France \\	
				\small
				\{matha.deghel, ejder.bastug, mohamad.assaad, merouane.debbah\}@centralesupelec.fr
		}
		\thanks{This research has been supported by the ERC Starting Grant 305123 MORE (Advanced Mathematical Tools for Complex Network Engineering) and the project BESTCOM.}
}
\IEEEoverridecommandlockouts
\maketitle
\vspace{-100ex}
\IEEEpeerreviewmaketitle

\begin{abstract}
In this contribution, we jointly investigate the benefits of caching and interference alignment (IA) in multiple-input multiple-output (MIMO) interference channel under  limited backhaul capacity. In particular, total average transmission rate is derived as a function of various system parameters such as backhaul link capacity, cache size, number of active transmitter-receiver pairs as well as the quantization bits for channel state information (CSI). Given the fact that base stations are equipped both with caching and IA capabilities and have knowledge of content popularity profile, we then characterize an \emph{operational regime} where the caching is beneficial. Subsequently, we find the optimal number of transmitter-receiver pairs that maximizes the total average transmission rate. When the popularity profile of requested contents falls into the operational regime, it turns out that caching substantially improves the throughput as it mitigates the backhaul usage and allows IA methods to take benefit of such limited backhaul.
\end{abstract}

\section{Introduction}
The current mobile cellular networks are evolving towards $5G$ wireless networks, aiming to sustain the huge rise of connected devices and data-hungry application of mobile users. Among the possible solutions \cite{Andrews2014Will}, proactively caching users' contents at the network edge is shown to achieve significant gains in terms of users' satisfaction and offloading gains \cite{Bastug2014LivingOnTheEdge}. Specifically, the idea of caching is to smartly move the users' contents close to mobile users, yielding less access delays to the contents and reducing the backhaul usage. In the same context, one of the key issue in wireless communication systems is the interference which is caused by the large number of simultaneous transmissions on the same channel, resulting into severe performance degradations unless treated properly. In this regard, \ac{IA} is introduced as an efficient interference management method and is shown to result in higher throughputs compared to conventional interference-agnostic methods.

In the context of cellular networks, caching was recently studied  by different research groups, both in terms of gains and approximation algorithms \cite{Bastug2014CacheEnabled, 
Blaszczyszyn2014Geographic,
Hamidouche2014ManyToMany,
Pantisano2014MatchToCache,
Karamchandani2014Hier, 
Golrezaei2013FemtocachingD2D,
Poularakis2014Approximation,
Paakkonen2014Regenerating, 
Liu2014CacheEnabled, 
Liu2014CacheInduced}. On the other hand, \ac{IA} was initially introduced in \cite{Cadambe2008Interference}, and is shown to achieve maximum multiplexing gain in \ac{MIMO} channels \cite{Gou2010Degrees} under the assumption that all the transmitters have perfect global \ac{CSI}. In \ac{FDD} systems, the imperfect case with \ac{CSI} quantization process for single-antenna receivers \cite{Bolcskei2009Interference}, and multiple-antenna receivers \cite{Krishnamachari2013Interference, Chen2014Performance} are studied, showing that the \ac{DoF} can be achieved at high \ac{SNR} regime by using a specific quantization scheme with optimal number of feedback bits. The \ac{IA} methods that exploit channel reciprocity in \ac{TDD} systems are studied (see \cite{Rezaee2013Csit, Chang2009Optimal, Chaporkar2008Optimal, Chaporkar2009Scheduling} for instance), assuming that the \ac{CSI} acquisition cost is independent of the transmission rate and is linear in the number of probed receivers. 
In fact, most of aforementioned \ac{IA} methods rely on \ac{CSI} exchange over the backhaul links and do not consider the  implications of data traffic on the limited backhaul links and exchange process. From these observations, one can bring caching into the scenario as a way of creating opportunities for \ac{CSI} exchange over the backhaul. In other words, \ac{IA} methods could have higher throughputs as the amount of data traffic over the backhaul is substantially reduced, since
this reduction results in a saved capacity which can be used for the CSI sharing process.

Based on the motivations above, the main contribution of this work is to jointly analyze the benefits of caching and \ac{IA} methods under the limited backhaul. In particular, given the fact that users' content requests follow a certain popularity profile (i.e., few contents might be highly popular than the rest or all might have similar popularities), we aim to find an operational regime where the caching is beneficial to \ac{IA} methods in terms of throughput. To show this, we first derive the expressions for average throughput, then characterize this regime based on the  shape of content popularity profile. Finally, we maximize the total average throughput as a key metric of interest. In a similar vein, the work in \cite{Liu2014CacheEnabled} has jointly studied the caching and power control problem for opportunistic cooperative \ac{MIMO}. Therein, closed form expressions for power control are derived based on approximated Bellman equation and convex stochastic caching problem is solved via a stochastic subgradient algorithm. The proposed scheme is shown to be asymptotically optimal in the high \ac{SNR} regime. Another joint solution for cooperative \ac{MIMO} was introduced in \cite{Liu2014CacheInduced}, where both caching control and the optimal \ac{MIMO} precoder in transmit power minimization are investigated.

The rest of this paper is structured as follows. Our system model is given in Section \ref{sec:systemmodel}, including the details of the \ac{MIMO} interference channel model, \ac{IA} and caching capabilities at the transmitters with limited backhaul. In Section \ref{sec:performance}, the expressions for average transmission rate are derived as the main performance metrics. Based on these expressions, an operational caching regime that meets certain \ac{QoS} criteria is provided by relying on content popularity profile. Then, an optimization problem for maximizing the average transmission rate is formulated, where the number of active transmitter-receiver pairs is optimized subject to the backhaul capacity constraints. Section \ref{sec:numerical} is dedicated to numerical results and relevant discussions. We finally conclude and draw our future directions in Section \ref{sec:conclusions}.

\emph{Notation}: Boldface uppercase symbols (i.e., ${\bf B}$) represent matrices whereas lowercases (i.e., ${\bf b}$) are used for vectors. The symbol ${\bf I}$ denotes square identity matrix. $(.)^*$ denotes the conjugate transpose. $|.|$ indicates the absolute value and $||.||$ is used for the norm of second degree. Lastly, $\mathcal{CN}({\bf b}, {\bf B})$ corresponds to a complex Gaussian random vector with mean ${\bf b}$ and covariance matrix ${\bf B}$.
\section{System Model}
\label{sec:systemmodel}
\begin{figure}[ht!]
	\centering
	\includegraphics[width=0.48\linewidth]{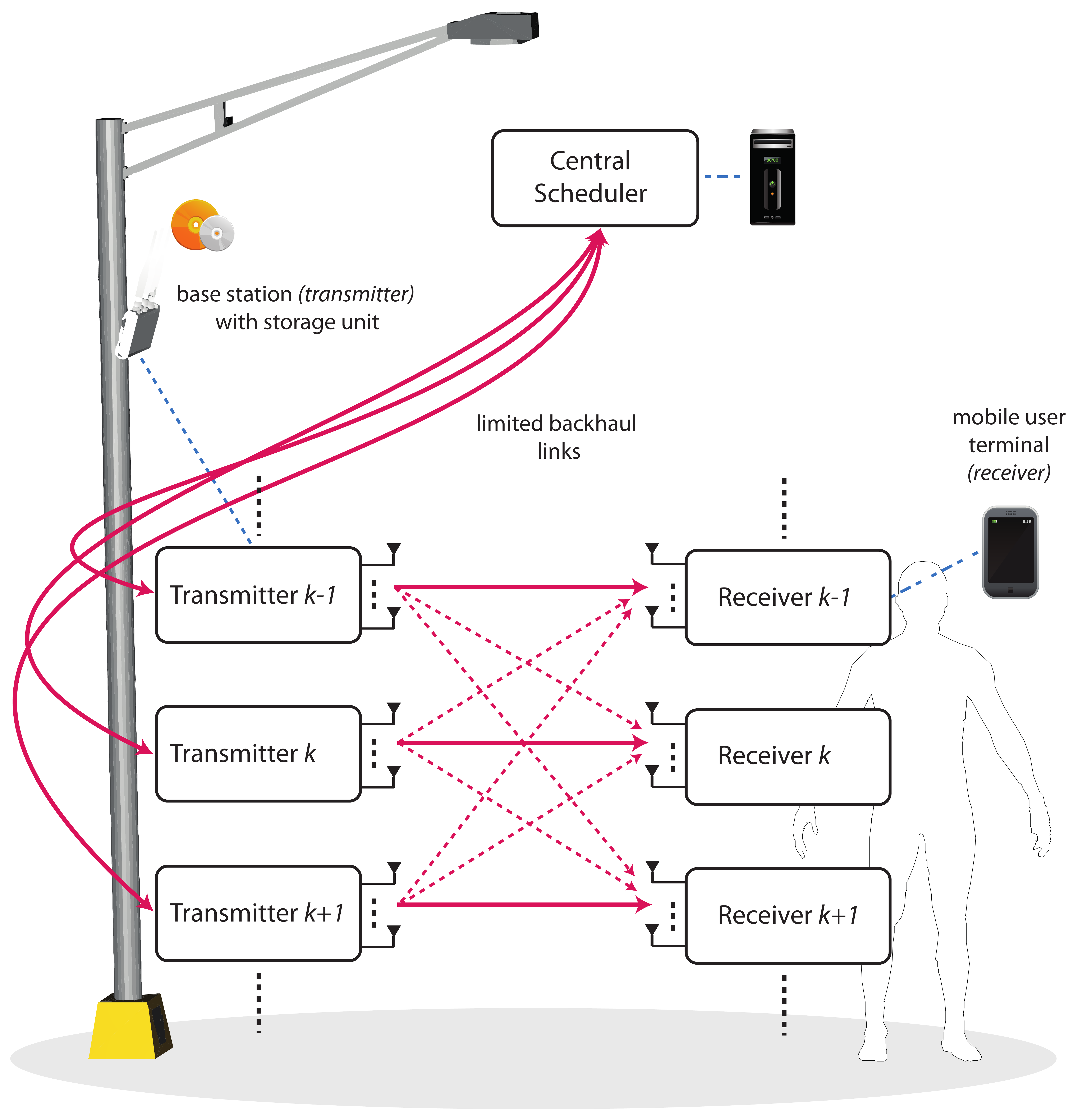}
	\caption{A sketch of $L$-User \ac{MIMO} interference network.}
	\label{fig:scenario}
\end{figure} 
We consider a \ac{MIMO} interference channel with $L$ transmitter-receiver pairs, as illustrated in Fig. \ref{fig:scenario}. For simplicity, we assume a homogeneous network where all transmitters (base stations) are equipped with $N_{\mathrm{t}}$ antennas and all receivers (users) with $N_{\mathrm{r}}$ antennas. The number of independent data streams from transmitter $k$ to its paired receiver $k$ is denoted by $d_k$, with $d_k \le \min(N_{\mathrm{t}},N_{\mathrm{r}})$.

Given this \ac{MIMO} interference channel model, the received signal at user $k$ can be written as
\begin{equation}
\mathbf{y}_k=\sum\limits_{i=1}^{L} \sqrt{ \frac{\zeta_{ki}P}{d_i} } \mathbf{H}_{ki} \sum\limits_{j=1}^{d_i}  \mathbf{v}_i^j x_i^j + \mathbf{z}_k \label{eq:yk}
\end{equation} 
where $\mathbf{y}_k$ is the $N_{\mathrm{r}} \times 1$ received signal vector, $\mathbf{H}_{ki}$ is the $N_{\mathrm{r}} \times N_{\mathrm{t}}$ channel matrix between transmitter $i$ and receiver $k$ with i.i.d. $ \mathcal{CN} (0,1)$ elements, $\zeta_{ki}$ represents the path loss of channel $ \mathbf{H}_{ki}$, $P$ is the total power at each transmitter equally allocated among its streams, $x_i^j $  denotes the $j$-th data stream from transmitter $i$, $\mathbf{v}_i^j \in \mathbb{C}^{N_{\mathrm{t}} \times 1}$ is the corresponding precoding vector of unit norm and $\mathbf{z}_k$ is a vector of i.i.d. complex Gaussian noise with covariance matrix  $\sigma^2\mathbf{I}_{N_{\mathrm{r}}}$. We denote by $\alpha_{ki}$ the fraction $\frac{\zeta_{ki}P}{d_i}$, for all $k, i$ in $\{1,...,L\}$.
%
\subsection{Interference Alignment}
\ac{IA} is a linear precoding technique which can be adopted for the \ac{MIMO} interference channel. While this technique is commonly used with multiple receiver design, for the sake of simplicity we restrict ourselves to a per-stream zero-forcing receiver. Specifically, let receiver $k$ use the \emph{combiner} vector $\mathbf{u}_k^m \in \mathbb{C}^{N_{\mathrm{r}} \times 1}$ of unit norm to detect the $m$-th stream from transmitter $k$, such as 
\begin{flalign}
\hat{x}_k^m  \nonumber &= \left(\mathbf{u}_k^m \right)^* \mathbf{y}_k &\\
\nonumber &=   \overbrace{  \sqrt{\alpha}_{kk}  \left(\mathbf{u}_k^m \right)^* \mathbf{H}_{kk} \mathbf{v}_k^m x_k^m   }^{\text{desired signal}}     +  \overbrace{  \sqrt{\alpha}_{kk} \sum\limits_{\substack{j=1 \\ j\ne m}}^{d_k} \left(\mathbf{u}_k^m \right)^* \mathbf{H}_{kk} \mathbf{v}_k^j x_k^j }^{\text{inter-stream interference (\ac{ISI})}}   &\\  &\qquad {} +   \overbrace{   \sum\limits_{\substack{i=1 \\ i\ne k}}^{L}  \sqrt{\alpha}_{ki}  \sum\limits_{j=1}^{d_i} \left(\mathbf{u}_k^m \right)^* \mathbf{H}_{ki} \mathbf{v}_i^j x_i^j  }^{\text{inter-user interference (\ac{IUI})}} + \overbrace{  \left(\mathbf{u}_k^m \right)^* \mathbf{z}_k}^{\text{noise}}.   \label{eq:xk} 
\end{flalign} 
As observed from \eqref{eq:xk}, two sources of interference affect the detection of the stream at the receiver, namely i) the \ac{ISI} and ii) the \ac{IUI}. The \ac{IA} technique is used to manage this problem by designing the set of precoder and combiner vectors such that
\begin{align}
\left(\mathbf{u}_k^m \right)^* \mathbf{H}_{ki} \mathbf{v}_i^j = 0,  &&  \forall (k,m)\ne (i,j).  \label{eq:IAeq} 
\end{align} 
The \emph{perfect interference alignment} is achieved if the above conditions hold. In other words, supposing that perfect global \ac{CSI} is available at all the transmitters and each receiver consequently obtains a perfect version of the combiner vector designed at its corresponding transmitter, \ac{IUI} and \ac{ISI} can be  canceled completely at the receivers. It turns out that obtaining the perfect global \ac{CSI} at the transmitters is not a straightforward task in practice due to the limited backhaul. The \ac{CSI} sharing mechanism over the limited backhaul is detailed in the following. 
%
\subsection{CSIT Sharing Over Limited Capacity Backhaul Links}
As alluded earlier, global \ac{CSI} is required at each transmitting node in order to design the \ac{IA} vectors that satisfy \eqref{eq:IAeq}. As shown in Fig. \ref{fig:scenario}, we suppose that all the transmitters are connected to a central node via  their limited backhaul links, which serves as: (i) a way for connecting transmitters to each other and (ii) a mean to link the system to the Internet for data transfer. We assume a \ac{TDD} transmission strategy where the users send their training sequences, allowing each transmitter to estimate its \emph{local} \ac{CSI}, meaning that the $i$-th transmitter estimates perfectly the channels  $\mathbf{H}_{ki}$, $k=1,...,L$. However, the local \ac{CSI} (excluding the direct links) of other transmitters are obtained via backhaul links of limited capacity. 

In this contribution, we suppose that the backhaul is error-free and has a fixed capacity of $C$. The capacity of each link from a transmitter to the central node is then given by $C_k=\frac{C}{L}$, as a function of the number of active transmitter-receiver pairs. Note that $k$ refers to pair $k$, where $k=1,...,L$. Denoting $C_{\mathrm{c}}$ as the capacity reserved for \ac{CSI} sharing and $C_{\mathrm{d}}$ as the part dedicated to data transfer, the capacity of each link can be also written as $C_k = C_{k\mathrm{c}}+C_{k\mathrm{d}}$. We assume that $C_{k\mathrm{c}}=\frac{C_{\mathrm{c}}}{L}$ and $C_{k\mathrm{d}}=\frac{C_{\mathrm{d}}}{L}$. In such limited backhaul conditions, a codebook-based quantization technique needs to be adopted to reduce the huge amount of information exchange used for \ac{CSI} sharing, which we detail as follows. Let $\mathbf{h}_{ki} $ denote the vectorization of the channel matrix $\mathbf{H}_{ki}$. Then, for all $k \ne i$, transmitter $i$ selects the index $n_o$ which corresponds to the optimal codeword in a predetermined codebook $\mathcal{CB}=\left[\mathbf{\hat {h}}_{ki}^1,...,\mathbf{\hat {h}}_{ki}^{2^B} \right] $ according to
\begin{equation}
 n_o=  \operatorname*{arg\,max}_{1\le n \le 2^B} \left| \mathbf{\tilde{h}}_{ki}^* \mathbf{\hat {h}}_{ki}^n  \right|^2, \label{eq:n0}
\end{equation}
in which $B$ is the number of bits used to quantize $\mathbf{H}_{ki}$ and  $\mathbf{\tilde{h}}_{ki} =  \frac{\mathbf{h}_{ki}}{\left\| \mathbf{h}_{ki}  \right\|}$ is the channel direction vector.

After quantizing all the matrices of its local \ac{CSI}, we assume that transmitter $i$ sends the corresponding optimal indexes to all other transmitters which share the same codebook, allowing these transmitters to reconstruct the quantized local knowledge of transmitter $i$. Let us now define the quantization error as $e_{ki}= 1-\frac{ \left| \mathbf{\hat{h}^*}_{ki} \mathbf{h}_{ki} \right|^2 }  { \left\| \mathbf{h}_{ki} \right\|^2 } $ and adopt the same model in \cite{Jindal2007UserSel}, \cite{Lau2012Stability} which relies on the theory of quantization cell approximation. The \ac{CDF} of $e_{ki}$ is then given by
 \begin{align}
\text{Pr}(e_{ki} \le \varepsilon) =
  \begin{cases}
   2^B \varepsilon^Q, &  0\le \varepsilon \le 2^{- \frac{B}{Q}} \\
   1, & \varepsilon > 2^{- \frac{B}{Q}}
  \end{cases}
  \end{align}
 where $Q=N_{\mathrm{t}}N_{\mathrm{r}}-1$.

Recall that we consider a finite capacity backhaul in which we perform a quantization scheme to reduce the \ac{CSI} sharing cost. Since these limited capacity backhaul links are also used for actual data transfer, one additional way to allocate more capacity for \ac{CSI} sharing is to decrease this data transfer. This is generally accomplished by means of caching in which we describe in the following.
%
\subsection{Cache-enabled Transmitters}
Several studies have shown that certain types of content are relatively more requested than others such as viral videos with millions of views, share of popular people in social media, well-known news and blog pages. Indeed, accessing the same information by many users is one of the major reasons for network congestion and latency increase. Let us assume that each transmitter is associated with a storage unit (cache) which stores the content with respect to a certain popularity profile. 

At the transmitters, for ease of analysis, we consider the trivial approach that consists in storing the most popular content, which results from the reasonable fact that a user's request  matches with the global popular contents \cite{Bastug2014CacheEnabled}. Indeed, the content popularity can be described by the probability distribution function, given by the following expression
\begin{align}
f_\text{pop}(f, \eta) =
  \begin{cases}
   (\eta-1)f^{- \eta}, &  f \ge 1 \\
    0, & f< 1   
  \end{cases}  \label{eq:fpop}
  \end{align}
where $f$ represents a point in the support of the corresponding content, and $\eta$ stands for a factor that describes the steepness of the popularity distribution curve. Lower values of $\eta$ corresponds to a uniform behaviour (almost all contents have the same popularities), whereas a high $\eta$ value would results in a steeper distribution (very few contents are highly popular than the rest). Now, suppose that each transmitter stores the contents up to $f_0$ (namely cache size) from the distribution in \eqref{eq:fpop}. Then, the probability that a content request falls in the range $\Delta=[0,f_0]$, namely \emph{cache hit probability}, can be calculated as  
\begin{align}
\text{Pr}_{\text{hit}} \nonumber &= \int_{0}^{f_0} f_\text{pop}(f, \eta) \, df \\ &=  1-f_0^{1-\eta}.    \label{eq:Phit}
\end{align} 
Consequently, the probability that a content demand is missing from the cache can be given by $\text{Pr}_{\text{miss}} = 1- \text{Pr}_{\text{hit}} = f_0^{1-\eta}$. Based on the above model which considers \ac{IA} and caching capabilities at the transmitters, we next focus on the performance analysis of the system.
%
\section{Performance Analysis}
\label{sec:performance}
In this section, we derive the expression for the total average transmission rate and characterize an operational regime where caching is beneficial. Then, we provide an optimization problem that maximizes the transmission rate.
%
\subsection{Average Transmission Rate}
As explained in the preceding section, the \ac{IA} vectors are designed based on the available \ac{CSI} that results after the transmitting nodes quantize and share their perfect local knowledge between each other. Thus, the IA technique adopted is able to completely suppress the \ac{ISI} since local \ac{CSI} is perfectly known, but not the \ac{IUI} because of the quantization process which leads to imperfect global \ac{CSI} at the transmitters. 
Under such conditions and using the results in \cite{Deghel2014System}, the \ac{SINR} for stream $m$ at receiver $k$ can be expressed as
 \begin{align}
\gamma_k^m \nonumber &=  \frac{ \alpha_{kk} \left|  \left(\mathbf{\hat{u}}_k^m \right)^* \mathbf{H}_{kk} \mathbf{\hat{v}}_k^m  \right|^2  }{  \sigma^2 +   \sum\limits_{\substack{i=1 \\ i \ne k}}^{L}   \alpha_{ki} \sum\limits_{j=1}^{d_i} \left| \left(\mathbf{\hat{u}}_k^m \right)^* \mathbf{H}_{ki} \mathbf{\hat{v}}_i^j \right|^2  } \nonumber \\ &=
 \frac{ \alpha_{kk} \left|  \left(\mathbf{\hat{u}}_k^m \right)^* \mathbf{H}_{kk} \mathbf{\hat{v}}_k^m  \right|^2  }{  \sigma^2 +  \sum\limits_{\substack{i=1 \\ i \ne k}}^{L} \alpha_{ki} \left\| \mathbf{h}_{ki} \right\|^2  e_{ki}  \sum\limits_{j=1}^{d_i}  \left|  \mathbf{w}_{ki}^* \mathbf{s}_{k,i}^{m,j} \right|^2  }, \label{eq:SINR}
 \end{align}
 where $\mathbf{w}_{ki}$ is a unit norm vector isotropically distributed in the null space of $\mathbf{\hat{h}}_{ki}$, $\mathbf{s}_{k,i}^{m,j} = \mathbf{\hat{v}}_i^j \otimes (\mathbf{\hat{u}}_k^m)^* $ ($\otimes$ is the Kronecker product), $\mathbf{\hat{v}}_k^m$ and $\mathbf{\hat{u}}_k^m$ are the precoding and combining vectors, respectively, designed based on the available \ac{CSI} described in the previous section. 

Using the \ac{SINR} expression in $\eqref{eq:SINR}$, the \emph{instantaneous rate} for user $k$ can be given by 
\begin{align}
R_k  \nonumber &= \\ & \sum\limits_{m=1}^{d_k}  \log_2 \left(1+ \frac{ \alpha_{kk} \left|  \left(\mathbf{\hat{u}}_k^m \right)^* \mathbf{H}_{kk} \mathbf{\hat{v}}_k^m  \right|^2  }{ \sigma^2 +   \sum\limits_{\substack{i=1 \\ i \ne k}}^{L}   \alpha_{ki} \sum\limits_{j=1}^{d_i} \left| \left(\mathbf{\hat{u}}_k^m \right)^* \mathbf{H}_{ki} \mathbf{\hat{v}}_i^j \right|^2  }  \right).  \label{eq:Rk}
\end{align} 

We assume that the quantization error plays the role of an additional source of Gaussian noise, regardless of its distribution \cite{El2012Overhead}. Under this assumption, the \emph{average rate} for user $k$ achieved by \ac{IA} can be written as 
\begin{align}
&\bar{R}_k  \nonumber = \\ & \sum\limits_{m=1}^{d_k} \mathbb{E} \left[ \log_2 \left(1+ \frac{ \alpha_{kk} \left|  \left(\mathbf{\hat{u}}_k^m \right)^* \mathbf{H}_{kk} \mathbf{\hat{v}}_k^m  \right|^2  }{   \sigma^2 + \sum\limits_{\substack{i=1 \\ i \ne k}}^{L}   \alpha_{ki} \sum\limits_{j=1}^{d_i} \mathbb{E}  \left[  \left| \left(\mathbf{\hat{u}}_k^m \right)^* \mathbf{H}_{ki} \mathbf{\hat{v}}_i^j \right|^2 \right]  }  \right) \right] \label{eq:AvRk}
\end{align} 
where we note that the outer expectation is only over the direct channel. Therefore, the leakage interference terms $ \left(\mathbf{\hat{u}}_k^m \right)^* \mathbf{H}_{ki}  \mathbf{\hat{v}}_i^j $ are nothing but an independent sources of additive Gaussian noise, irrespective of their actual distribution. The following lemma will be useful for the rest of analysis.
\begin{lemma}
The average rate for user $k$ can be written in exponential form as 
 \label{pro:AveRk}
	\begin{align}
	\bar{R}_k  =  d_k \log_2(e) e^{\frac{1}{ \beta_k }}  E_1 \left( \frac{1}{ \beta_k }  \right)
	 \end{align} 
 where $\beta_k= \frac{  P \zeta_{kk}} {  d_{kk} \left( \sigma^2 +  P 2^{1-\frac{B}{Q} } \sum\limits_{i=1, i \ne k}^{L}  \zeta_{ki}  \right) }  $ and $E_1(.)$ is the exponential integral defined as $E_1(a)=\int\limits_{1}^{\infty} t^{-1} e^{-at} dt $.
\label{eq:AveRk}
\end{lemma}
\begin{proof}
	The proof is provided in Appendix \ref{app:AveRk}.
\end{proof}

Note that the rate metrics we derived so far are related to the wireless downlink transmission achieved by \ac{IA}, whereas in the following, we shall  derive more elaborated rate expressions by taking into account caching and limited backhaul aspects. We shall now define the \emph{instantaneous transmission rate} for user $k$, such as
\begin{align}
 r_k =
   \begin{cases}
    R_k , &    f_r \in \Delta  \\
    C_{k\mathrm{d}}, &  f_r \notin \Delta \\
   \end{cases}     \label{eq:r}
\end{align}
where $f_r$ represents the requested content and $\Delta$ is the available catalog in the local cache. The main intuition behind this definition is the following. If the requested content exists in the local cache, the amount of rate given to the user is $R_k$. On the other hand, if the content does not exist in the local cache, the content is fetched from the Internet via the backhaul, thus the given rate is $C_{k\mathrm{d}}$. We assume that $C_{k\mathrm{d}} < R_k$ always holds. This assumption comes from the motivation that the backhaul link capacity in $5$G networks is expected to be a limited factor compared to wireless link capacity, especially in ultra-dense deployment of \glspl{BS} \cite{Andrews2014Will}. Given this definition and assumption, we state the following theorem.
\begin{theorem}[Average Transmission Rate] The average transmission rate for user $k$ can be given by
	\label{pro:avgrate}
	\begin{align}
	    \bar{r}_k  =   d_k \log_2(e) e^{\frac{1}{ \beta_k }}  E_1 \left( \frac{1}{ \beta_k }  \right) (1-f_0^{1-\eta_k}) +C_{k\mathrm{d}} f_0^{1-\eta_k}.     \label{eq:avgrate}
	\end{align} 
\end{theorem}
\begin{proof}
We have $\bar{r}_k= \mathbb{E}[R_k] \text{Pr}_{\text{hit}}+C_{k\mathrm{d}}\text{Pr}_{\text{miss}}  = \bar{R}_k (1-f_0^{1-\eta_k})+C_{k\mathrm{d}} f_0^{1-\eta_k}$. By replacing $\bar{R}_k$ by its expression given in Lemma \ref{pro:AveRk}, the result in \eqref{eq:avgrate} follows.
\end{proof}

Consequently, the \emph{total average transmission rate} of the system can be found straightforwardly by taking the sum over all the pairs of the expression in \eqref{eq:avgrate} as follows
\begin{align}
\bar{r}_{T}  \nonumber &= \\ & \sum\limits_{k=1}^{L} \left( d_k \log_2(e) e^{\frac{1}{ \beta_k }}  E_1 \left( \frac{1}{ \beta_k }  \right) (1-f_0^{1-\eta_k}) +C_{k\mathrm{d}} f_0^{1-\eta_k} \right)    \label{eq:Tavgrate}  
\end{align} 
\begin{remark}\label{remark:avgrate}
The more storage (caching) capacity increases, the more missing probability decreases, and consequently the hitting probability increases. Thus, for a fixed steepness factor $\eta$, the support of cached contents (represented by $f_0$) has an important impact on the total average transmission rate. Similar remarks can be given for the number of active pairs $L$ and the number of bits $B$.
\end{remark}
%
\subsection{Operational Caching Regime}
\label{sec:operational}
The steepness factor $\eta$ describes how much steep is the popularity distribution function, and it depends on requested contents of the corresponding user. In other words, a high value of $\eta$ results from the fact that some contents are much more popular than other contents and thus, because the cache contains the most popular contents, the hitting probability will be high. On the other side, a low value of $\eta$ is due to (more or less) the same popularity of the requested contents and then the hitting probability can not reach important values.
This analysis can be resumed by the following proposition.
\begin{proposition}
	\label{pro:rk(eta)}
	The average rate for user $k$ (with $k=1,...,L$) is an increasing function with respect to its corresponding steepness factor $\eta_k$.
\end{proposition}
\begin{proof}
	The first derivative $\frac{d\bar{r}_k}{d\eta_k}=(\bar{R}_k -C_{k\mathrm{d}}) f_0^{1-\eta_k} \ln f_0$. This derivative is positive since we have $\bar{R}_k > C_{k\mathrm{d}}$, and hence the statement of Proposition \ref{pro:rk(eta)} follows. 
\end{proof}
We will now derive two bounds based on the steepness factor $\eta_k$ of pair $k$, under different observations and constraints on the average transmission rate:
\subsubsection{Minimum Guaranteed Transmission Rate} 
A minimum desired average transmission rate at user $k$ can be expressed using the following inequality $\bar{r}_k \ge p\bar{R}_k$,  where $p<1$ is a \ac{QoS} factor that dictates how much the actual transmission rate should be achieved. Using this inequality, we can derive a lower bound on $\eta_k$ as
\begin{align}
\bar{r}_k=\bar{R}_k (1-f_0^{1-\eta_k})+C_{k\mathrm{d}} f_0^{1-\eta_k} \ge p\bar{R}_k,
\end{align}
thus results in a steepness factor
\begin{align}
\eta_k \ge 1-\frac{\ln \left( \frac{\bar{R}_k(1-p)}{ \bar{R}_k-C_{k\mathrm{d}}  } \right)  } { \ln f_0 }.
\end{align}
\subsubsection{Constant Average Rate Variation}
One could notice that there exists a regime where the average transmission rate has almost a constant variation in function of $\eta_k$. To detect this regime, a simple but effective way is to consider $\frac{d\bar{r}_k}{d\eta_k} < \epsilon$, where $\epsilon$ is a parameter that describes how much the first derivative is close to zero.
Under this consideration, we can calculate a lower bound on $\eta_k$ as
\begin{align}
\frac{d\bar{r}_k}{d\eta_k}=f_0^{1-\eta_k}(\bar{R}_k - C_{k\mathrm{d}}) \ln f_0 < \epsilon,
\end{align}
thus gives a steepness factor 
\begin{align}
\eta_k > 1-\frac{\ln \left( \frac{ \epsilon  }{  (\bar{R}_k-C_{k\mathrm{d}} )\ln f_0 } \right)  } { \ln f_0 }.
\end{align}

Let $\eta_{k1}=1-\frac{\ln \left( \frac{\bar{R}_k(1-p)}{ \bar{R}_k-C_{k\mathrm{d}}  } \right)  } { \ln f_0 }$ and $\eta_{k2}= 1-\frac{\ln \left( \frac{ \epsilon  }{  (\bar{R}_k-C_{k\mathrm{d}} )\ln f_0 } \right)  } { \ln f_0 } $.
Using these two bounds, we can define the regime where caching is beneficial for user $k$ in terms of average rate. Specifically, for a minimum guaranteed rate defined by $\bar{r}_k \ge p\bar{R}_k$ and for an average rate variation $\frac{d\bar{r}_k}{d\eta_k} \ge \epsilon$, caching is gainful for user $k$ (i.e. can satisfy these latter conditions) if its steepness factor is between these intervals, such as $\eta_{k1} \le \eta_k \le \eta_{k2}$. 
\subsection{Rate Maximization}
The total transmission rate in our setup is a function of various parameters. Among these parameters, we focus on the number of pairs $L$. We investigate the optimal value of $L$ by defining and solving an optimization problem which seeks to maximize the total average transmission rate. In fact, as it can be seen in \eqref{eq:Tavgrate}, solving this problem for the general case is of high complexity. Therefore, before proceeding in the definition of this optimization problem and for the sake of simplicity, we make the following assumptions: (i) all the transmitters have the same number of streams $d$, (ii) all the users have the same steepness factor denoted by $\eta$, and (iii) we use the extended Wyner model (1D system) where the path loss coefficient from transmitter $i$ to user $k$ is given by $\zeta^{|k-i|}$. We can represent this path loss model using the matrix
\begin{equation}
\bf A  =
  \begin{pmatrix}
   1 & \zeta & \zeta^2 & \cdots & \zeta^{L-1} \\
      \zeta & 1 & \zeta & \cdots & \zeta^{L-2} \\
      \zeta^2 & \zeta & 1 & \cdots & \zeta^{L-3} \\
      \vdots  & \vdots  & \vdots & \ddots&  \vdots  \\
     \zeta^{L-1} & \zeta^{L-2} & \zeta^{L-3}& \cdots & 1
  \end{pmatrix}.
\end{equation}
Under these assumptions and recalling that $C_{k\mathrm{d}}=\frac{C_{\mathrm{d}}}{L}$, we can re-express \eqref{eq:Tavgrate} as 
	\begin{equation}
	\bar{r}_{T_{\mathrm{s}}} = \begin{cases}
		 2a   \sum\limits_{i=1}^{\frac{L}{2}}e^{a_i}  E_1 \left( a_i   \right)   + b  & \text{if $L$ is even}\\
		2a   \sum\limits_{i=1}^{\lfloor \frac{L}{2} \rfloor}e^{a_i}  E_1 \left( a_i   \right)+ae^{b_1}  E_1 \left( b_1   \right)   + b &\text{ if $L$ is odd}
	\end{cases}
	\label{eq:simpavtr}
	\end{equation}
	where $a_i=d\sigma^2P^{-1}+ d2^{1-\frac{B}{Q}}(1-\zeta)^{-1} (2\zeta-\zeta^{L-i+1}-\zeta^i)$, $b_1=d\sigma^2P^{-1}+ d2^{1-\frac{B}{Q}}(1-\zeta)^{-1} 2(\zeta-\zeta^{\lfloor \frac{L}{2} \rfloor+1})$, $a=d \log_2(e) (1-f_0^{1-\eta})$, $b= C_{\mathrm{d}}   f_0^{1-\eta}$ and $\lfloor \frac{L}{2} \rfloor$ is the largest integer not greater than $\frac{L}{2}$.
\begin{remark}
	To ensure the feasibility of the \ac{IA}
	problem, the system parameters should satisfy the following condition (given in \cite{Yetis2010Feasibility}) $N_{\mathrm{t}}+N_{\mathrm{r}} \ge  d(L+1)$. Without loss of generality, we assume that the number of pairs $L$ satisfies this condition.
\end{remark}

Now, we can define our optimization problem which seeks to maximize the total average transmission rate in $\eqref{eq:simpavtr}$ with respect to the number of pairs $L$. This is formally stated as
\begin{align}
     \underset{L}{\text{maximize}}
      &\qquad  \bar{r}_{T_{\mathrm{s}}}(L)   \label{eq:optprob1}  \\
     \text{subject to}
      &\qquad  L^2(L-1)B \le  \left( C_{\mathrm{c}}+ (1-f_0^{1-\eta})C_{\mathrm{d}} \right)\tau \label{eq:consprob1}
\end{align} 
where $\tau$ is the slot duration. The term at the left hand side of \eqref{eq:consprob1} represents the total number of bits (needed for \ac{CSI} sharing) and is obtained from the fact that we have $L$ transmitters, each of which shares $L-1$ channels (using $LB$ bits for each channel)  to $L-1$ other transmitters. The right hand side of \eqref{eq:consprob1} shows how caching mitigates the backhaul usage, allowing higher capacity of backhaul links which are used for \ac{CSI} sharing. In detail, caching saves $(1-f_0^{1-\eta})C_{\mathrm{d}}$ of the backhaul capacity usage, and thus this saved part can be used, in addition to $C_{\mathrm{c}}$, in the \ac{CSI} sharing process. For the optimization problem, we first describe the behavior of $\bar{r}_{T_{\mathrm{s}}}$ in the following result.

\begin{proposition}
\label{pro:rtsBehavior}
The total average rate $ \bar{r}_{T_{\mathrm{s}}} $ is an increasing function with respect to the number of pairs $L$ (with $L\ge3$), for sufficiently small $\zeta$ values. 
\end{proposition}
\begin{proof}
 	The proof is provided in Appendix \ref{app:rtsBehavior}.
\end{proof}
Using the above proposition, the optimal number of pairs (denoted by $L_{\mathrm{opt}}$) can be easily obtained by setting $L=3$ and increasing it until condition $\eqref{eq:consprob1}$ is not satisfied. Note that Proposition \ref{pro:rtsBehavior} holds for sufficiently small values of $\zeta$. To solve the optimization problem for arbitrary $\zeta$ values ($\zeta < 1$), we use the following procedure.
\begin{enumerate}
   \item[] {\bf Step 1}: Compute  $\bar{r}_{T_{\mathrm{s}}}$ for all $L$ that satisfy conditions $\eqref{eq:consprob1}$ and $ d(L+1)\le N_t+N_r$.
   \item[] {\bf Step 2}: Select the maximum among the computed $\bar{r}_{T_{\mathrm{s}}}$ values and take the corresponding $L$ as $L_{\mathrm{opt}}$.
 \end{enumerate}
Notice that for a fixed number of pairs $L$, the same analysis can be done for the number of bits $B$. Using the condition in $\eqref{eq:consprob1}$ and since $\bar{r}_{T_{\mathrm{s}}}$ is an increasing function with $B$, an increase of bound  $C_{\mathrm{c}}+(1-f_0^{1-\eta})C_{\mathrm{d}}$ allows us to use more number of bits for the quantization process, and thus to get better total average rate $\bar{r}_{T_{\mathrm{s}}}$.

%

\section{Numerical Results}
\label{sec:numerical}
In this section we present our numerical results to validate the analysis conducted in the previous section. For ease of exposition, we consider a setup with $N_{\mathrm{t}}=N_{\mathrm{r}}=15$, $\text{SNR}=10\log_{10}\left( \frac{P}{\sigma^2} \right)=10$ dB, $d =2$, $\zeta=0.3$, $\tau=1$ ms,  $C_{\mathrm{d}}=5$ Mb/s and bandwidth $BW=$ 10 MHz per transmitter.

In Fig. \ref{fig:rFuncL} we plot the variation of the total average transmission rate with respect to the number of active pairs $L$. It can be seen that $\bar{r}_{T_{\mathrm{s}}}$ can be significantly increased by increasing the size of the catalog in transmitters, namely $f_0$. Furthermore, the impact of increasing the number of bits $B$ is higher for larger $f_0$. 
\begin{figure}[!ht]
\centering
\begin{tikzpicture}[scale=0.9]
	\begin{axis}[
 		grid = major,
 		legend cell align=left,
 		mark repeat={4},
 		xmin=3,xmax=14,	
 		legend style ={legend pos=north west},
 		xlabel={Number of pairs $L$},
 		ylabel={Total Avg. Transmission Rate $\bar{r}_{T_{\mathrm{s}}}$[Mb/s]},
 		cycle list name = laneasIACaching1]

 		\addplot [color=cyan!70!black,mark=none, thick] table [col sep=comma] {\string"rfLf010B30.csv"};
 		\addlegendentry{$B = 30$ bits};
 		
 		\addplot [color=red!70!black,dashed,mark=none, thick] table [col sep=comma] {\string"rfLf010B10.csv"};
 		\addlegendentry{$B = 10$ bits};

 	    \addplot [color=cyan!70!black,mark=none, thick] table [col sep=comma] {\string"rfLf0100B30.csv"};
 	    \addplot [color=red!70!black,dashed,mark=none, thick] table [col sep=comma] {\string"rfLf0100B10.csv"};
 	     	
 	    \addplot [color=cyan!70!black,mark=none, thick] table [col sep=comma] {\string"rfLf01000B30.csv"};
 	    \addplot [color=red!70!black,dashed,mark=none, thick] table [col sep=comma] {\string"rfLf01000B10.csv"};
 	     
        \node at (axis cs:8,26.2) [anchor=south west, rotate = 33] {\small $f_0=1000$};
 	    \node at (axis cs:8,17.2) [anchor=south west, rotate = 25] {\small $f_0=100$};
 	    \node at (axis cs:7.9,11.4) [anchor=south west, rotate = 15] {\small $f_0=10$};
	\end{axis}
\end{tikzpicture}
\caption{$\bar{r}_{T_{\mathrm{s}}}$ vs. $L$, with $\eta=1.2$.}
\label{fig:rFuncL}
\end{figure}
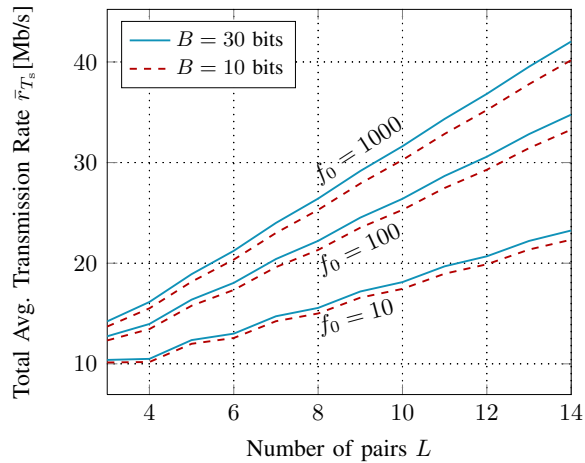

The evolution of average transmission rate with respect to the steepness factor is depicted in Fig. \ref{fig:rkFeta}. By looking into the feasible values of $\bar{r}_{k}$ in which $\eta_{k}$ is between $\eta_{k1}$ and $\eta_{k2}$ (recall Section \ref{sec:operational}), we can notice that $\bar{r}_{k}$ increases more dramatically as the size of catalog increases. Additionally, keeping aside the fact that the transmission rate is not guaranteed below $\eta_{k1}$, the variations after $\eta_{k2}$ are almost constant regardless of different catalog sizes. This confirms our expressions derived for the operational caching regime.
\begin{figure}[!ht]
\centering
\begin{tikzpicture}[scale=0.9]
	\begin{axis}[
 		grid = major,
 		legend cell align=left,
 		mark repeat={1},
 		xmin=1,xmax=4,	
 		legend style ={legend pos=south east},
 		xlabel={Steepness Factor $\eta_k$},
 		ylabel={Avg. Transmission Rate $\bar{r}_{k}$ [Mb/s]},
 		cycle list name = laneasIACaching1]
	
			\addplot [color=gray!70!black, dashed, mark=none, thick, forget plot]  table [col sep=comma] {\string"rFetaf0=1000-a.csv"};
			\addplot [color=red!70!black, mark=star,mark size=1.8, thick, mark repeat=8]  table [col sep=comma] {\string"rFetaf0=1000-b.csv"};
	 		\addlegendentry{$f_0=1000$}; 			
			\addplot [color=gray!70!black, dashed, mark=none, thick, forget plot]  table [col sep=comma] {\string"rFetaf0=1000-c.csv"};

			\addplot [color=gray!70!black, dashed, mark=none, thick, forget plot] table [col sep=comma] {\string"rFetaf0=100-a.csv"};
			\addplot [color=yellow!70!black, mark=square,mark size=1.8, thick, mark repeat=8] table [col sep=comma] {\string"rFetaf0=100-b.csv"};
			\addlegendentry{$f_0=100$};			
			\addplot [color=gray!70!black, dashed, mark=none, thick, forget plot] table [col sep=comma] {\string"rFetaf0=100-c.csv"};

			\addplot [color=gray!70!black, dashed, mark=none, thick, forget plot]  table [col sep=comma] {\string"rFetaf0=10-a.csv"};
			\addplot [color=cyan!70!black, mark=o,mark size=1.8, thick, mark repeat=8]  table [col sep=comma] {\string"rFetaf0=10-b.csv"};
 			\addlegendentry{$f_0=10$}; 
			\addplot [color=gray!70!black, dashed, mark=none, thick, forget plot]  table [col sep=comma] {\string"rFetaf0=10-c.csv"};
			 			
 			\node[pin={[pin distance=0.90cm]273:{$(\eta_{k1}, \bar{r}_{k})$ }}] at (axis cs: 1.445,6.05) {};
			\node[pin={[pin distance=1.37cm]320:{ }}] at (axis cs:1.20,5.9268) {}; 	
			\node[pin={[pin distance=1.54cm]325:{ }}] at (axis cs:1.105,5.8968) {};
			
			\node[pin={[pin distance=1.91cm]211:{ }}] at (axis cs:3.56,8.54) {}; 		
			\node[pin={[pin distance=1.08cm]295:{$(\eta_{k2}, \bar{r}_{k})$}}] at (axis cs:2.33,8.54) {};	
			\node[pin={[pin distance=2.10cm]331:{ }}] at (axis cs:1.89,8.54) {};
 	\end{axis}
\end{tikzpicture}
\caption{$\bar{r}_{k}$ vs. $\eta_k$, with $L=8$, $B=30$ bits, $p=0.7$ and $\epsilon=0.05$.}
\label{fig:rkFeta}
\end{figure}
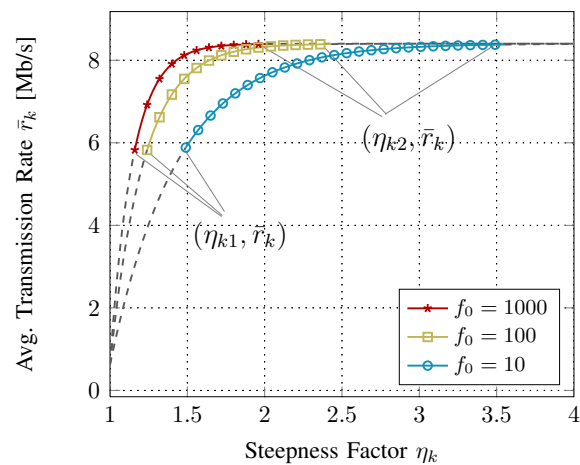	

The impact of steepness factor on the maximum total average rate is shown in Fig. \ref{fig:rTmaxFeta} for different values of the backhaul capacity dedicated to the \ac{CSI} sharing (namely $C_{\mathrm{c}}$). Given the fact that maximum total average rate is achieved by finding the optimal number of pairs $L_{\mathrm{opt}}$, improvement of this rate for a specific range of $\eta$ (as in operational caching regime) can be further fueled by increasing $C_{\mathrm{c}}$ and/or $f_0$. This behaviour in fact validates our analysis.
\begin{figure}[!ht]
\centering
\begin{tikzpicture}[scale=0.9]
	\begin{axis}[
 		grid = major,
 		legend cell align=left,
 		mark repeat={2},
 		xmin=1,xmax=4,	
 		legend style ={legend pos=south east, font=\tiny},
 		xlabel={Steepness Factor $\eta$},
 		ylabel={Max. Total Avg. Trans. Rate $\bar{r}_{T_{\mathrm{s}}}$ [Mb/s]},
 		cycle list name = laneasIACaching1]
	
		\addplot [color=cyan!70!black, thick, forget plot]  table [col sep=comma] {\string"rtmaxFetaf01000Cc20000.csv"};
	 	\addplot [color=cyan!70!black, thick, forget plot] table [col sep=comma] {\string"rtmaxFetaf0100Cc20000.csv"};   	
  	   	\addplot [color=cyan!70!black, thick]  table [col sep=comma] {\string"rtmaxFetaf010Cc20000.csv"};
 		\addlegendentry{$C_{\mathrm{c}}=20$ Mb/s};

 		\addplot [dashed,color=red!70!black, thick, forget plot]  table [col sep=comma] {\string"rtmaxFetaf01000Cc2000.csv"};
 		\addplot [dashed,color=red!70!black, thick, forget plot] table [col sep=comma] {\string"rtmaxFetaf0100Cc2000.csv"}; 	      	
 		\addplot [dashed,color=red!70!black, thick]  table [col sep=comma] {\string"rtmaxFetaf010Cc2000.csv"};
 		\addlegendentry{$C_{\mathrm{c}}=2$ Mb/s};

		\node[pin={[pin distance=2.4cm]25:{}}] at (axis cs: 1.082,23.8) {};
		\node[pin={[pin distance=1.31cm]25:{}}] at (axis cs: 1.505,25.8) {};
		\node at (axis cs:2.772,31.1) [anchor=east] {\small $f_0=1000$};
				
		\node[pin={[pin distance=2.46cm]315:{}}] at (axis cs: 1.33,32.2) {};
		\node[pin={[pin distance=1.456cm]315:{}}] at (axis cs:  1.61,26.45) {};
		\node at (axis cs:2.72,18.3) [anchor=east] {\small $f_0=100$};
		
		\node[pin={[pin distance=2.24cm]315:{}}] at (axis cs: 1.35,25.494) {};
		\node[pin={[pin distance=1.58cm]315:{}}] at (axis cs: 1.51,21.65) {};
		\node at (axis cs:2.62,12.9) [anchor=east] {\small $f_0=10$};
 	\end{axis}
\end{tikzpicture}
\caption{Maximum $\bar{r}_{T_{\mathrm{s}}}$ vs. $\eta$, with $B=30$ bits.}
\label{fig:rTmaxFeta}
\end{figure}
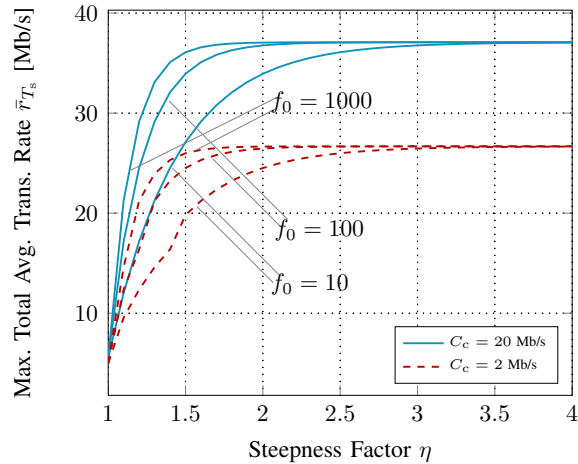	

Fig. \ref{fig:Lopt} illustrates the variation of $L_{\mathrm{opt}}$ with respect to the capacity $C_{\mathrm{c}}$, for different values of steepness factor $\eta$. It can be noticed that, for the same $\eta$, $L_{\mathrm{opt}}$ increases with $C_{\mathrm{c}}$ and can reach larger values for higher steepness factor $\eta$. Recall that $L_{\mathrm{opt}}$ also depends on the capacity $C_{\mathrm{d}}$ and the cache size $f_0$ (see the bound in $\eqref{eq:consprob1}$).
\begin{figure}[!ht]
\centering
\begin{tikzpicture}[scale=0.9]
	\begin{axis}[
 		grid = major,
 		legend cell align=left,
 		mark repeat={4},
 		xmin=0.5,xmax=20,	
 		legend style ={legend pos=south east},
 		xlabel={Capacity $C_{\mathrm{c}}$ [Mb/s]},
 		ylabel={Optimal Nr. of Pairs $L_{\mathrm{opt}}$},
 		cycle list name = laneasIACaching1]

 		\addplot [color=red!70!black,mark=star,mark size=1.8, thick] table [col sep=comma] {\string"Lopteta3f0=10.csv"};
		\addlegendentry{$\eta = 3.0$};

 		\addplot [color=yellow!70!black,mark=square,mark size=1.8, thick] table [col sep=comma] {\string"Lopteta1.5f0=10.csv"};
   	    \addlegendentry{$\eta = 1.5$};
   	    
 		\addplot [color=cyan!70!black,mark=o,mark size=1.8, thick]  table [col sep=comma] {\string"Lopteta1.1f0=10.csv"};
 		\addlegendentry{$\eta = 1.1$};
	\end{axis}
\end{tikzpicture}
\caption{$L_{\mathrm{opt}}$ vs. $C_{\mathrm{c}}$, with $B=30$ bits and $f_0=10$.}
\label{fig:Lopt}
\end{figure}
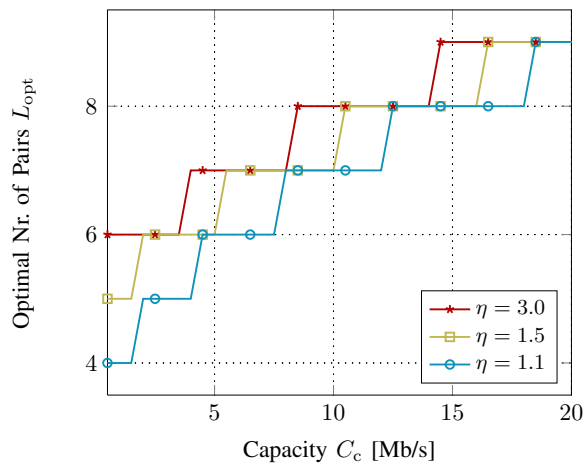

%

\section{Conclusion}
\label{sec:conclusions}

In this paper, we have analyzed the performance of the interference alignment technique applied to a $L$-user \ac{MIMO} system, under the limited backhaul capacity and caching capabilities at the transmitters. Under some specific assumptions and considerations, we derived expressions of the total average transmission rate $\bar{r}_{T_{\mathrm{s}}}$ and the operational caching regime has been determined based on the content popularity profile. A key observation of this work is that, under this regime, cache-enabled base stations can significantly increase the $\bar{r}_{T_{\mathrm{s}}}$ as compared to traditional \glspl{BS}. We also showed the existence of an optimum number of pairs for the total average rate, and that this optimum number depends on several parameters such as capacity $C_{\mathrm{c}}$, steepness factor $\eta$ and storage size $f_0$. 

The implication of caching in wireless networks is of high interest and requires further investigations. For instance, solving the optimization problems for the general case would be an interesting result. In addition, the impact of caching on other interference management techniques can be investigated. Lastly, heterogeneous network scenarios, including macro cells and small cells deployments, can be added as an additional layer to reveal the benefits of caching and \ac{IA} methods for future networks.
\bibliographystyle{IEEEtran}
\bibliography{references}
%
\appendices
\section{Proof of Lemma \ref{pro:AveRk}}
\label{app:AveRk}

We start by calculating the inner expectation in \eqref{eq:AvRk} given by $\mathbb{E}  \left[  \left| \left(\mathbf{\hat{u}}_k^m \right)^* \mathbf{H}_{ki} \mathbf{\hat{v}}_i^j \right|^2 \right] $. 
From \eqref{eq:SINR}, we have the following:
$\mathbb{E}\left| \left(\mathbf{\hat{u}}_k^m \right)^* \mathbf{H}_{ki} \mathbf{\hat{v}}_i^j \right|^2 = \mathbb{E} \left[\left\| \mathbf{h}_{ki} \right\|^2  e_{ki}    \left|  \mathbf{w}_{ki} \mathbf{s}_{k,i}^{m,j} \right|^2 \right] $. According to \cite[Appendix A]{Chen2014Performance}, $\left\| \mathbf{h}_{ki} \right\|^2  e_{ki}    \left|  \mathbf{w}_{ki} \mathbf{s}_{k,i}^{m,j} \right|^2$ is equal to $2^{-\frac{B}{N_{\mathrm{t}}N_{\mathrm{r}}-1}}  \chi^2(2)  $ in distribution. Since $\chi^2(2)$ has a mean equal to  $2$, then we have $\mathbb{E}  \left[  \left| \left(\mathbf{\hat{u}}_k^m \right)^* \mathbf{H}_{ki} \mathbf{\hat{v}}_i^j \right|^2 \right] = 2^{1-\frac{B}{N_{\mathrm{t}}N_{\mathrm{r}}-1}}= 2^{1-\frac{B}{Q}}$.
Thus, the expression in \eqref{eq:AvRk} can be re-expressed as the following:
\begin{align}
\bar{R}_k  \nonumber &=  \sum\limits_{m=1}^{d_k} \mathbb{E} \left[ \log_2 \left(1+ \frac{ \alpha_{kk} \left|  \left(\mathbf{\hat{u}}_k^m \right)^* \mathbf{H}_{kk} \mathbf{\hat{v}}_k^m  \right|^2  }{  \sigma^2 + \sum\limits_{\substack{i=1 \\ i \ne k}}^{L}   \alpha_{ki} d_i   2^{1-\frac{B}{Q}}   }  \right) \right] \\ &=   \sum\limits_{m=1}^{d_k} \mathbb{E} \left[ \log_2 \left(1+ \frac{ P \zeta_{kk} \left|  \left(\mathbf{\hat{u}}_k^m \right)^* \mathbf{H}_{kk} \mathbf{\hat{v}}_k^m  \right|^2  }{    d_k( \sigma^2 +  P 2^{1-\frac{B}{Q}} \sum\limits_{\substack{i=1 \\ i \ne k}}^{L}   \zeta_{ki}  )}  \right) \right].
\end{align} 

We now need to calculate the outer expectation. For this, we use the result in \cite{El2012Overhead}:\\
$\mathbb{E} \left[ \log_2 \left(1+ \frac{P \zeta_{kk} \left|  \left(\mathbf{\hat{u}}_k^m \right)^* \mathbf{H}_{kk} \mathbf{\hat{v}}_k^m  \right|^2  }{ d_k\sigma_k^2} \right) \right]= \log_2(e) e^{\frac{1}{ \beta_k }}  E_1 \left( \frac{1}{ \beta_k }  \right) $, where $\sigma_k^2= \sigma^2 +   P 2^{1-\frac{B}{Q}} \sum\limits_{i=1, i \ne k}^{L}   \zeta_{ki} $, $\beta_k= \frac{P\zeta_{kk}}{ d_k\sigma_k^2}$ and $E_1(.)$ is the exponential integral function.  Therefore, the average rate for user $k$ can be given by 
\begin{align}
	\bar{R}_k  \nonumber &=  \sum\limits_{m=1}^{d_k} \log_2(e) e^{\frac{1}{ \beta_k }}  E_1 \left( \frac{1}{ \beta_k }  \right) \\ &= d_k \log_2(e) e^{\frac{1}{ \beta_k }}  E_1 \left( \frac{1}{ \beta_k }  \right).
\end{align} 
This concludes the proof.\hfill$\blacksquare$
\section{Proof of Proposition \ref{pro:rtsBehavior}}
\label{app:rtsBehavior}
We recall that  $\bar{r}_{T_{\mathrm{s}}}$ is given by the following
	\begin{equation}
	\bar{r}_{T_{\mathrm{s}}}  = \begin{cases}
		 2a   \sum\limits_{i=1}^{\frac{L}{2}}e^{a_i}  E_1 \left( a_i   \right)   + b  & \text{if $L$ is even}\\
		2a   \sum\limits_{i=1}^{\lfloor \frac{L}{2} \rfloor}e^{a_i}  E_1 \left( a_i   \right)+ae^{b_1}  E_1 \left( b_1   \right)   + b &\text{ if $L$ is odd}
	\end{cases}
	\label{eq:simpavtrproof}
	\end{equation}
	where $a_i=d\sigma^2P^{-1}+ d2^{1-\frac{B}{Q}}(1-\zeta)^{-1} (2\zeta-\zeta^{L-i+1}-\zeta^i)$, $b_1=d\sigma^2P^{-1}+ d2^{1-\frac{B}{Q}}(1-\zeta)^{-1} 2(\zeta-\zeta^{\lfloor \frac{L}{2} \rfloor+1})$, $a=d \log_2(e) (1-f_0^{1-\eta})$ and $b= C_{\mathrm{d}}   f_0^{1-\eta}$. For sufficiently small values of $\zeta$, we can suppose that $2\zeta+2\zeta^2+2\zeta^3+\cdots \approx 2\zeta$, or equivalently $\zeta+\zeta^2+\zeta^3+\cdots \approx\zeta$. To justify this, take for instance $\zeta=0.1$ which yields $0.1+0.1^2+0.1^3+\cdots=0.11\approx 0.1$.
	
Consequently, we get $(1-\zeta)^{-1} (\zeta-\zeta^{L-i+1})=\zeta+\cdots+\zeta^{L-i}\approx \zeta$,  $(1-\zeta)^{-1}(\zeta-\zeta^{\lfloor \frac{L}{2} \rfloor+1})=\zeta+\cdots+\zeta^{       \lfloor \frac{L}{2} \rfloor}\approx \zeta  $ and also $(1-\zeta)^{-1} (\zeta-\zeta^{i}) \approx \zeta$ (for $i>1$). Therefore, the expression in \eqref{eq:simpavtrproof} simplifies to 
\begin{align}
	\bar{r}_{T_{\mathrm{s}}} \approx 
		 2ae^{c_1}  E_1 \left( c_1   \right) + (L-2)ae^{c_2}  E_1 \left( c_2   \right)  + b,
	\label{eq:simpavtrproof1}
	\end{align}
where $c_1=d\sigma^2P^{-1}+ d2^{1-\frac{B}{Q}}\zeta$ and $c_2=d\sigma^2P^{-1}+ d2^{1-\frac{B}{Q}}2\zeta$. Based on expression \eqref{eq:simpavtrproof1}, we conclude that the total average rate $\bar{r}_{T_{\mathrm{s}}}$ is linear with the number of pairs $L$. Hence, the desired result holds. \hfill$\blacksquare$

\end{document}